\documentclass[twocolumn,pra,aps,showpacs,superscriptaddress,nofootinbib]{revtex4-1}
\usepackage{amsmath,amsthm,amssymb}
\usepackage{graphicx}
\usepackage{color}
\makeatletter
\def\be{\begin{equation}}
\def\ee{\end{equation}}
\def\bea{\begin{eqnarray}}
\def\eea{\end{eqnarray}}
\def\ben{\begin{equation*}}
\def\een{\end{equation*}}
\def\bean{\begin{eqnarray*}}
\def\eean{\end{eqnarray*}}
\def\bma{\begin{mathletters}}
\def\ema{\end{mathletters}}
\def\bi{\begin{itemize}}
\def\ei{\end{itemize}}
\newtheorem{thm}{Theorem}

\newtheorem{lem}{Lemma}[thm]
\newcommand{\ket}[1]{ | \, #1 \rangle}

\begin{document}
\title{True Multipartite Entanglement Hardy Test}

\author{Ramij Rahaman}\email{ramijrahaman@gmail.com}
\affiliation{Department of Mathematics, University of Allahabad, Allahabad 211002, U.P., India}
\affiliation{Institute of
Theoretical Physics and Astrophysics, University of Gda\'{n}sk,
80-952 Gda\'{n}sk, Poland}
\author{Marcin Wie\'{s}niak}
\affiliation{Institute of Theoretical Physics and Astrophysics,
University of Gda\'{n}sk, 80-952 Gda\'{n}sk, Poland}
\author{Marek \.{Z}ukowski}
\affiliation{Institute of
Theoretical Physics and Astrophysics, University of Gda\'{n}sk,
80-952 Gda\'{n}sk, Poland}


\begin{abstract}
Quantum mechanics allows systems to be entangled with each other, which results in stronger than classical correlations. Many methods of identifying entanglement have been proposed over years, most of which	are based on violating some statistical inequalities. In this work we extend the idea due to Hardy, in which entanglement is not identified with use of statistical inequalities, but by simultaneous satisfaction of certain conditions. We show that the new variant of the Hardy paradox relying on marginal probabilities can be resolved only by true multipartite entangled states. Also, the state resolving this paradox for given local measurements is pure and unique in case of qubit collections.
\end{abstract}
\maketitle
\section{Introduction}
Fifty years ago Bell proved \cite{Bel64} that one can find measurement correlations for a composite quantum system which cannot be described by any local realistic theory (LRT). Such theories can also be called local hidden variable ones (LHVT), or local causal, see e.g. \cite{CASLAV}. The approach of Bell was statistical. Bell's inequalities, in fact, are statistical predictions about some sets of measurements which can be made on particles far separated from each other. A direct contradiction between quantum mechanics and local realism was found in 1989 by Greenberger, Horne and Zeilinger (GHZ) \cite{GHZ89}. In their argumentation they used correlations of a state of four spin-$\frac{1}{2}$ particles $\frac{1}{\sqrt{2}}(\ket{0000}-\ket{1111})$, and remarked that for the three-qubit analog of the state their thesis holds too. Although their proof is direct, it requires at least the eight-dimensional Hilbert space. It works only for the aforementioned states, in contrast to Bell inequalities like \cite{CHSH69}, which can be violated by any pure state and a wide range of mixed ones.

The structure of multipartite entanglement is not a simple extension of the bipartite one. For example, for three qubits there are two different classes of pure genuinely three-partite entanglements, and also one may have entanglement of just two parties. Most of features of bipartite entanglement are well understood, whereas the multipartite entanglement this is still not the case \cite{Sev87,ZBLW02,LPZC04,SS02,CGPRS02,BBGL11,BPBG11, BSV12, CGKR10,YCZLO12}. The rich structure of the multipartite entanglement can be used for various tasks, such as quantum computation \cite{RB01}, quantum simulation \cite{Llo96}, quantum metrology \cite{PS09}. This inspired broad theoretical and experimental studies, \cite{HOR, PAN12}.

In 1992, Hardy \cite{Har92} gave a proof of a no-go theorem for local hidden variables which requires only two qubits and does not require inequalities. We extend the approach of Hardy to more complicated situations, and show that Hardy-type conditions for correlations precisely determine a specific {\em genuine multipartite entangled state}\footnote{The state is not-biseparable with respect to any partition of subsystems.}, which can satisfy them.

By refereing to marginal probabilities, we introduce a novel test of true multipartite entanglement based on Hardy correlations. We demonstrate that our Hardy conditions are satisfied by a unique pure state, which opens many ways to new potential applications in quantum information processing tasks.
\section{Hardy's paradox.} Consider a physical system consisting of
two subsystems shared between two distant parties Alice and Bob. The two observers
(Alice and Bob) have access to one subsystem each. Assume that
Alice can execute experiments measuring one (chosen at
random) of the two dichotomic ($\pm 1$-valued) variables $\hat{u}_1$ and $\hat{v}_1$, and Bob can do the same with
variables $\hat{u}_2$ and $\hat{v}_2$, also chosen at random, and of a similar nature.

The Hardy-type argument \cite{Har92} starts with the following set of four joint probability
conditions:
\be\label{hardy2q}\begin{split}
P(+,+|\hat{u}_1,\hat{u}_2) &= q>0,\\
P(+,+|\hat{v}_1,\hat{u}_2) &= 0,\\
P(+,+|\hat{u}_1,\hat{v}_2) &= 0,\\
P(-,-|\hat{v}_1,\hat{v}_2) &= 0,
\end{split}\ee
where $P(j,k|X_1,X_2)$ denotes the joint probability of observing result, with $j,k=\pm 1$ under local setting $X_i\in \{\hat{u}_i,\hat{v}_i\}$, and we use the convention $\pm$ to denote $\pm1$. This set of conditions cannot be satisfied by any LRT \cite{Har92}. Proof of this follows immediately by showing a
violation by the above probabilities of the following inequality, which holds for local realistic models
\be\label{CH}
\begin{split}
&P(+,+|\hat{v}_1,\hat{u}_2)+P(+,+|\hat{u}_1,\hat{v}_2)\\+&P(-,-|\hat{v}_1,\hat{v}_2)-P(+,+|\hat{u}_1,\hat{u}_2)\geq 0.
\end{split}\ee
The inequality is equivalent to the famous CH inequality \cite{CH74}, or rather its left hand side:
\ben\begin{split}
-1&\leq P(+,+|\hat{v}_1,\hat{u}_2)+P(+,+|\hat{u}_1,\hat{v}_2)+P(+,+|\hat{v}_1,\hat{v}_2)\\&-P(+|\hat{v}_1)-P(+|\hat{v}_2)-P(+,+|\hat{u}_1,\hat{u}_2)\leq 0.\end{split}
\een
The equivalence is reached immediately, once one notices that \ben P(-,-|\hat{v}_1,\hat{v}_2)= P(+,+|\hat{v}_1,\hat{v}_2)-P(+|\hat{v}_1)-P(+|\hat{v}_2) +1.\een
From conditions (\ref{hardy2q}) in Eq. (\ref{CH}) we have $0>-q> 0$.
This contradiction proves that Hardy conditions cannot be satisfied by any local realistic model.
\subsection{Generalization of Hardy's argument} Let us present a generalized form of Hardy's argument. Consider $N$ subsystems shared among $N$ separated parties. Assume that $i$-th party can measure one of two observables, $\hat{u}_i$ and $\hat{v}_i$, on the local subsystem. The outcomes $x_i$ of each such measurement can be $1,2,...,d_i$. Here $d_i$ is the dimension of Hilbert space associated to the $i$-th subsystem. We now consider all the joint probabilities $P(x_1,x_2,...,x_N|\hat{x}_1, \hat{x}_2,...,\hat{x}_N)$, where $\hat{x}_i \in \{\hat{u}_i, \hat{v}_i\}$. A Hardy-type argument \cite{Har92} can start from the following set of conditions:
\begin{widetext}
\begin{eqnarray}\label{HardyG}
& P(1,1,.....,1|\hat{u}_1,\hat{u}_2,.....,\hat{u}_N) = q >0,&\nonumber\\
&\forall r, v_r\neq d_r: P(1,..,1,v_r,1,..,1|\hat{u}_1,..,\hat{u}_{r-1},\hat{v}_r,\hat{u}_{r+1},..,\hat{u}_N) = 0,&\\
&P(d_1,d_2,......,d_N|\hat{v}_1,\hat{v}_2,.....,\hat{v}_N) = 0.&\nonumber
\end{eqnarray}
\end{widetext}
Here, we use the convention $N+1\equiv 1$. This set of conditions cannot be satisfied by any LHVT.

To see this explicitly, let $\lambda$ be a local hidden variable (LHV), fully describing the entire system, taking values from a set $\Omega$ and $\rho(\lambda)$ be the complete state description for the joint system. In a LHVT description there exists conditional probabilities $p(u_j|\hat{u}_j,\lambda)$, $p(v_j|\hat{v}_j,\lambda)$, such that
$P(x_1,x_2,...,x_N|\hat{x}_1, \hat{x}_2,...,\hat{x}_N)= \int_{\lambda \in \Omega} d\lambda\rho(\lambda)~\prod_{j=1}^N p(x_j|\hat{x}_j,\lambda)$, where $\hat{x}_j\in \{\hat{u}_j,\hat{v}_j\}$. Thus, from the first condition in (\ref{HardyG}) we see that there exists a hidden variable subset of $\Omega$ of a non-zero measure, say $\Omega'$, within which for all $i$ one has $p(1|\hat{u}_i,\lambda)\neq 0$, and additionally $\rho(\lambda)\neq 0$. Now the second condition in (\ref{HardyG}) provides us for all $r$, $p(v_r|\hat{v}_r, \lambda)=0$ for all $v_r\neq d_r$ and for all $\lambda$'s in $\Omega'$. As one must have $\sum_{v_r=1}^{d_r}p(v_r|\hat{v}_r,\lambda)=1$, this immediately implies that $p(d_r|\hat{v}_r,\lambda)=1$ for all $\lambda \in \Omega'$. Therefore,
\ben\begin{array} {l}P(d_1,d_2,...,d_N|\hat{v}_1,\hat{v}_2,...,\hat{v}_N)\\= \int_{\lambda \in \Omega} \prod_{r=1}^{N}p(d_r|\hat{v}_r,\lambda)~\rho(\lambda)~d\lambda\\\geq \int_{\lambda \in \Omega'} \prod_{r=1}^{N}p(d_r|\hat{v}_r,\lambda)~\rho(\lambda)~d\lambda\\ = \int_{\lambda \in \Omega'} \rho(\lambda) d\lambda > 0,\end{array}\een
which is in contradiction with the last condition from set (\ref{HardyG}). Hence, conditions (\ref{HardyG}) cannot hold for LHVT. A similar proof is also given in \cite{GKS98} for a three-spins-$\frac{1}{2}$ system.
\section{Modified Hardy's conditions for the general case} One can modify the above conditions (\ref{HardyG}), in a way which we present below.
Consider the following set of joint probability conditions:
\begin{eqnarray}\label{HardyMGG}
& P(1,1,....,1|\hat{u}_1,\hat{u}_2,....,\hat{u}_N) = q >0,&\nonumber\\
& \forall r \mbox{ and $v_r\neq d_r$}:~ P(v_r,1|\hat{v}_r,\hat{u}_{r+1}) = 0,&\\
&P(d_1,d_2,....,d_N|\hat{v}_1,\hat{v}_2,....,\hat{v}_N) = 0,&\nonumber
\end{eqnarray}
where $P(v_r,1|\hat{v}_r,\hat{u}_{r+1})$ denotes the marginal. The new conditions cannot be satisfied by any LHVT.
The proof of this is similar to the previous one. Consider a LHVT as above. From the first condition in (\ref{HardyMGG}) we see that there exists a value range ($\Omega''$, say) of $\Omega$ within which, for all $r$, all the probabilities $p(1|\hat{u}_r,\lambda)$ and $\rho(\lambda)$
are all non-zero. The second condition from (\ref{HardyMGG}) provides us for all $r$, $p(v_r|\hat{v}_r,\lambda)=0$
for all $\lambda$'s in $\Omega''$ and for all $v_r\neq d_r$. This immediately implies that $p(d_r|\hat{v}_r,\lambda)=1$
for all $\lambda \in \Omega''$. Therefore, $P(d_1,d_2,..,d_N|\hat{v}_1,\hat{v}_2,..,\hat{v}_N)>0$, which contradicts the last condition in (\ref{HardyMGG}).

Let us start with introducing the central theorem of this work.
\begin{thm}\label{genuine}
Only a genuine multipartite entangled state satisfies the modified Hardy-type conditions (\ref{HardyMGG}).
\end{thm}
\begin{proof}
Consider state $\rho$ satisfying conditions (\ref{HardyMGG}), which is not genuinely $N$-partite entangled. Thus it is a convex combination of (at least) bi-separable states. Each of them is bi-separable with respect to some cut, say, $(1,2,...,m)$ vs. $(m+1,m+2,...,N)$. As all Hardy conditions are expressed in terms of probabilities, there must be at least one term in the convex combination which gives a non-zero contribution to the first condition, with $q>0$. Assume that such a term has the following form $Q(x_1,...,x_m|\hat{x}_1,...,\hat{x}_m)R(x_{m+1},...,x_N|\hat{x}_{m+1},...,\hat{x}_N)$. All Hardy conditions must hold for this term. One must have: by the first one
\begin{equation}
Q(1,...,1|\hat{u}_1,...,\hat{u}_m)R(1,...,1|\hat{u}_{m+1},...,\hat{u}_N)=q'>0, \label{relation1}
\end{equation}
whereas the middle conditions imply two important relations \ben\begin{split} & Q(v_m\neq d_m|\hat{v}_m)R(1|\hat{u}_{m+1}) =0, \\ \mbox{ and }& R(v_N\neq d_N|\hat{v}_N)Q(1|\hat{u}_1)=0.\end{split}\een

The structure of $Q$: Since by equation (\ref{relation1}) one must have $R(1|\hat{u}_{m+1})\neq 0$, we get $Q(v_m\neq d_m |\hat{v}_m)=0$, which implies that $Q(d_m|\hat{v}_m)=1$. Thus the state of the $m$-th particle in the considered term is pure, and is the eigenstate of operator $\hat{v}_m$ associated with the eigenvalue $d_m$. Therefore
$Q(x_1,...,x_m|\hat{x}_1,...,\hat{x}_m)$ factorizes to certain $Q'(x_1,...,x_{m-1}|\hat{x}_1,...,\hat{x}_m)Q_m(x_m|\hat{x}_m)$. The mid Hardy condition implies also that $Q'(v_{m-1}\neq d_{m-1}|\hat{v}_{m-1})Q_m(1|\hat{u}_m)=0$, however again since $Q_m(1|\hat{u}_m)\neq 0$ one has $Q'(d_{m-1}|\hat{v}_{m-1})=1$.
We have a next factorization of $Q$. Proceeding like this we show that $Q$ fully factorizes.

The structure of $R$: we have $R(v_N\neq d_N|\hat{v}_N)Q(1|\hat{u}_1)=0$, form which follows, in the same way as before that $R(v_N\neq d_N|\hat{v}_N)=0$, thus $R(d_N|\hat{v}_N)=1$. Following the same steps as above, we can show full factorization of $R$.

Thus the state representing the considered term is fully factorizable. Such states admit local hidden variable models, and as such cannot satisfy the modified Hardy condition for $q'>0$. Hence we reach a contradiction with
condition (\ref{relation1}). Since the proof is analogous for all cuts, no mixture of bi-separable states with respect to different cuts can satisfy all Hardy conditions.
\end{proof}
\subsection{Construction of state satisfying conditions (\ref{HardyMGG})}
Can one pinpoint a class of states which satisfy conditions (\ref{HardyMGG}), for specific pairs of local observables? For this purpose we will use a commonly known method, described in Ref. \cite{PG11}. Let us denote the eigenstates of $\hat{u}_j$ and $\hat{v}_j$ as $\ket{u_j}$ and $\ket{v_j}$, respectively, where $u_j,v_j$ denote eigenvalues. Let us now look for all the n-partite product states $|{\phi}_k\rangle = |\eta\rangle_1 |\eta\rangle_2 .......|\eta\rangle_N$, each of which is associated to the zero probabilities given in argument (\ref{HardyMGG}):
\bea\label{product} &|{\phi}_k(x_1,..,x_{r-1},v_r\neq d_r,u_{r+1}=1,x_{r+2},..,x_N)\rangle &\nonumber\\&\equiv\ket{x_1} ....\ket{x_{r-1}}\ket{v_r\neq d_r}\ket{u_{r+1}=1}\ket{x_{r+2}}....\ket{x_N}\\
&\mbox{and~~}|{\phi}_0\rangle
\equiv \ket{v_1 =d_1}\ket{v_2 =d_2} ........\ket{v_n =d_N},~~~~~~&\nonumber\eea
where $|x_l\rangle$ is any state of the $l$-th subsystem.
It is obvious that all the product states given in Eq. (\ref{product}) are not linearly independent. Let there be only $s$ linearly independent product states $\{\ket{\phi_i}\}_{i=1}^s$ of the form given in Eq. (\ref{product}). It is not very difficult to see that $\ket{\phi_0}$ is orthogonal to all the states given in Eq. (\ref{product}). Thus, states $\{\ket{\phi_{i}}\}_{i=0}^s$ are all linearly independent states and span a $(s+1)$-dim. subspace $\mathbb{S}$ of $\mathcal{H}_1^{d_1}\otimes \mathcal{H}_2^{d_2}\otimes.....\otimes{H}_N^{d_N}$. Here $s+1\leq d_1d_2.....d_N-1$, as $\ket{\phi}=\ket{u_1=1}\ket{u_2=1}....\ket{\hat{u}_N=1} \not\in \mathbb{S}$.

To satisfy the conditions given in Eqs. (\ref{HardyMGG}), a state $\rho$ has to be confined to the subspace of $\mathcal{H}_1^{d_1}\otimes \mathcal{H}_2^{d_2}\otimes.....\otimes{H}_N^{d_N}$, which is orthogonal to $\mathbb{S}$, call it a Hardy subspace $\mathbb{S}^{\perp}$. Thus, any state $\rho\in \mathbb{S}^{\perp}$ with $\langle \phi |\rho|\phi\rangle \neq 0$ will satisfy conditions (\ref{HardyMGG}). If one can show that $s$ can be big enough to have $s+1= d_1d_2.....d_N-1$, then the Hardy state $\rho$ is pure. Below we give examples for which this is the case.
\subsection{3-qubit Hardy-type state} Let us find the set of states $\rho$ for which the conditions for our Hardy-type argument given by conditions (\ref{HardyMGG}) are satisfied for a given set of three observable pairs $(\hat{u}_j, \hat{v}_j)$ ($j = 1, 2,3$). Take (for all $j = 1, 2,3$):
\begin{equation*}
\label{basischoice}
\begin{array}{lcl}
|\hat{u}_j = 1\rangle &=& \alpha_j|\hat{v}_j = 1\rangle +
\beta_j|\hat{v}_j
= 2\rangle,\\
|\hat{u}_j = 2\rangle &=& \beta_j^*|\hat{v}_j = 1\rangle -
\alpha_j^*|\hat{v}_j = 2\rangle,
\end{array}
\end{equation*}
where
$|\alpha_j|^2 + |\beta_j|^2 = 1$ and $0 < |\alpha_j|, |\beta_j|<1$.
The last condition is due
to the non-commutativity of $\hat{u}_j$ and $\hat{v}_j$. Linearly independent product states associated with the zero probabilities of conditions (\ref{HardyMGG}) are:
\ben
\begin{array}{lcl}
\ket{\phi_0}&=&\ket{v_1=2}\ket{v_2=2}\ket{v_3=2},\\
\ket{\phi_1}&=&\ket{v_1=1}\ket{u_2=1}\ket{u_3=1},\\
\ket{\phi_2}&=&\ket{v_1=1}\ket{u_2=1}\ket{u_3=2},\\
\ket{\phi_3}&=&\ket{u_1=1}\ket{v_2=1}\ket{u_3=1},\\
\ket{\phi_4}&=&\ket{u_1=2}\ket{v_2=1}\ket{u_3=1},\\
\ket{\phi_5}&=&\ket{u_1=1}\ket{u_2=1}\ket{v_3=1},\\
\ket{\phi_6}&=&\ket{u_1=1}\ket{u_2=2}\ket{v_3=1}.\end{array}\een
The product state associated with the first condition reads
$\ket{\phi_7}=\ket{u_1=1}\ket{u_2=1}\ket{u_3=1}$.

State $\rho$ that corresponds to conditions (\ref{HardyMGG}), has to be confined to a subspace of $\mathcal{C}^{2}\otimes \mathcal{C}^{2}\otimes{C}^{2}$, which is orthogonal to the subspace $\mathbb{S}=\{\ket{\phi_i}\}_{i=0}^6$. However, it's not orthogonal to the product state $\ket{\phi_7}$. The subspace $\mathbb{S}$ has dimension seven, so $\rho$ must be a pure genuine 3-qubit entangled state, which we denote as $\ket{\psi}$. As one can see, all the eight product states $\{\ket{\phi_i}\}_{i=0}^7$ are linearly independent, hence by using the Gram-Schmidt orthonormalization procedure one can find an orthonormal basis $\{\ket{\phi'_i}\}_{i=0}^7$, in which state $\ket{\psi}$ is its last member, with $i=7$:
\ben
\begin{array}{lcl}
\ket{\phi'_0} = \ket{\phi_0},~\ket{\phi'_i} = \frac{\ket{\phi_{i}}-\sum^{i-1}_{j=0}\langle \phi'_j|\phi_{i}\rangle\ket{\phi'_j}}
{\sqrt{1-\sum^{i-1}_{j=0}|\langle \phi'_j|\phi_{i}\rangle|^2}}, \mbox{~for $i=1,...,7$}.\\
\end{array}
\een
The probability $q$ in the conditions (\ref{HardyMGG}), for the Hardy state, reads
\ben\label{value_q}
q = |\langle \psi |\phi_7\rangle|^2 = 1-\sum_{i=0}^6|\langle \phi'_i|\phi_7\rangle|^2=
\frac{|\alpha_1\alpha_2\alpha_3|^2|\beta_1\beta_2\beta_3|^2}{1-|\alpha_1\alpha_2\alpha_3|^2}.
\een
Its maximum possible value is $0.0181938$. Further examples of Hardy states for bipartite cases can be found in \cite{PG11}.

For qubit systems the modified Hardy's conditions (\ref{HardyMGG}) can be expressed as
\begin{eqnarray}
&P(+,+,...,+|\hat{u}_1,\hat{u}_2,...,\hat{u}_N)=q>0,& \label{hardy1}\\
&\forall~r\leq N:~P(+,+|\hat{v}_r,\hat{u}_{r+1})=0,&\label{hardy2}\\
&P(-,-,...,-|\hat{v}_1,\hat{v}_2,...,\hat{v}_N)=0.& \label{hardy3}
\end{eqnarray}
\begin{lem}
Only a {\em unique} pure genuinely entangled $N$-qubit state satisfies (\ref{hardy1}-\ref{hardy3}).
\end{lem}
\begin{proof}
We show the uniqueness proof only, the proof for the genuine $N$-qubit entanglement follows from Theorem \ref{genuine}. Let us denote the eigenstates of $\hat{u}_j$ ($\hat{v}_j$) with eigenvalue $+1$ and $-1$ by $\ket{0_j}(\ket{+}_j)$ and $\ket{1_j}(\ket{-}_j)$, respectively. {\em Take $N$-qubit product states for which the measurements specified in Eqn. (\ref{hardy2}) give a probability equal to $1$, which belong to the following family:
\be\label{product1} \begin{split}&|{\phi}(x_1,..,x_2,+_r,0_{r+1},x_{r+2},..,x_N)\rangle \\& \equiv\ket{x_1}....\ket{x_{r-1}}\ket{+_r}\ket{0_{r+1}}\ket{x_{r+2}}....\ket{x_N},\end{split}\ee
where $\ket{x_j}\in\{\ket{0_j},\ket{1_j},\ket{+}_j,\ket{-}_j\}$, and the product state
\be\label{product2}|{\phi}_0\rangle
\equiv \ket{-}\ket{-} ........\ket{-}.\ee
for which $P(-,-,..,-|\hat{v}_1,\hat{v}_2,..,\hat{v}_N)=1.$ }
{\em Consider additionally a product state
\be\label{product3}
\ket{\phi_+}=\ket{0}\ket{0} ....\ket{0}....\ket{0} \mbox{ (or, simply $\ket{00...0...0}$)},
\ee
for which $P(+,+,..,+|\hat{u}_1,\hat{u}_2,..,\hat{u}_N)=1$, compare Eqn. (\ref{hardy1})}. Note that the product states given in Eqns. (\ref{product1}-\ref{product3}) are not linearly independent. Let's define a new product basis:
\begin{widetext}\ben\begin{split}
\ket{00...0...0}&=\ket{\phi_+}\\
\ket{00...01_l0...0}&=\frac{1}{\beta_l}\left[\ket{{\phi}_k(0,..,0,+_l,0,...,0)}-
\alpha_l\ket{\phi_+}\right],~\forall l,\\
\ket{0...01_l0...01_m0...0}&= \frac{1}{\beta_l\beta_m}\biggl[\ket{{\phi}_k(0,...,0,+_l,0,...,0,+_m,0,...,0)}-
\alpha_l\alpha_m\ket{\phi_+}\\&-\beta_l\alpha_m\ket{00...01_l0...0}
-\alpha_l\beta_m\ket{00...01_m0...0}\biggr],~\forall l\neq m,\\
\ket{0...01_l0...01_m0...01_k0...0}&=\frac{1}{\beta_l\beta_m\beta_k}\biggl[\ket{{\phi}_k(0,...,0,+_l,0,...,0,+_m,0,...,0,+_k,0,...,0)}-
\alpha_l\alpha_m\alpha_k\ket{\phi_+}
-\alpha_l\alpha_m\beta_k\ket{00...01_k0...0}\\&-\alpha_l\beta_m\alpha_k\ket{00...01_m0...0}
-\beta_l\alpha_m\alpha_k\ket{00...01_l0...0}
-\alpha_l\beta_m\beta_k\ket{00...01_m0...01_k0...0}\\&-\beta_l\alpha_m\beta_k\ket{00...01_l0...01_k0...0}
-\beta_l\beta_m\alpha_k\ket{00...01_l0...01_m0...0}\biggr],~\forall l\neq m\neq k\neq l,\\&
......\\
\ket{11...1...1}&=\frac{(-1)^N}{\prod_{i=1}^N{\alpha^*_i}}\Biggl[\ket{\phi_0}-\Biggl{\{}\left(\prod_{j=1}^N\beta^*_j\right)\ket{\phi_+}
+(-1)^1\sum_{i=1}^N\alpha^*_i\left(\prod_{j=1,j\neq i}^N\beta^*_j\right)\ket{00...01_i0...0}+\\
&+(-1)^2\sum_{i,l=1, i\neq l}^N\alpha^*_i\alpha^*_l\left(\prod_{j=1,j\neq i,l}^N\beta^*_j\right)\ket{00...01_i0...01_l0...0}+..........\\&+(-1)^{N-1}\sum_{j=1}^N\beta^*_j\left(\prod_{i=1,i\neq j}^N\alpha^*_i\right)\ket{11...10_j1...1}\Biggr{\}}\Biggr],
\end{split}\een\end{widetext}
where $\ket{+}_j=\alpha_j\ket{0}_j+\beta_j\ket{1}_j$, and $\ket{-}_j=\beta^*_j\ket{0}_j-\alpha^*_j\ket{1}_j$ with $|\alpha_j|^2+|\beta_j|^2=1$. {\em As we were able to construct the basis, product states given in Eqns. (\ref{product1}-\ref{product3}) must span the full $2^N$-dimensional Hilbert space. Let $\mathcal{S}_1=\{\ket{\phi (x_1,..,x_2,+_r,0_{r+1},x_{r+2},..,x_N)}\}$, be the subspace spanned by the product states given in Eqn. (\ref{product1}). Since $\ket{\phi_+}\not\perp \ket{\phi_0}$, whereas $\ket{\phi_0}\perp \mathcal{S}_1$, therefore, $\ket{\phi_+}\not\in \mathcal{S}_1$. However, since $\ket{\phi_0}$ and $\ket{\phi_+}$ are linearly independent and $\mathcal{S}_1\cup\{\ket{\phi_0},\ket{\phi_+}\}$ spans the full $2^N$-dimensional Hilbert space, therefore, $\mathcal{S}_1$ must be a $(2^N-2)$-dimensional, hence the dimension of the subspace $\mathcal{S}$, which it the one spanned by $\mathcal{S}_1\cup\{\ket{\phi_0}\}$ must be $2^N-1$. Thus, to satisfy the conditions (\ref{hardy1}-\ref{hardy3}) a state $\rho$ has to be orthogonal to the subspace $\mathcal S$. But, the subspace orthogonal to $\mathcal S$ is one dimensional. Thus, $\rho$ is a unique pure state. }
\end{proof}
Thus, an important feature of original Hardy-type two-qubit argument is preserved. This feature is missing in most of other multipartite Bell-type tests and totally absent in the case of generalized Hardy-type argument (\ref{HardyG}) for more than two-qubit case.

{\em Conclusions.} In summary, our modified Hardy-type test does not involve any statistical inequality. For $N$ qubit systems, we prove the uniqueness and purity of the Hardy state, and for the general case its genuine $N$-partite entanglement. Finally we remark that we also studied the maximum probability of success, $q$, of the modified Hardy-type test (\ref{HardyMGG}) for three two-level systems under a {\em generalized non-signaling theory (GNST)} and in quantum theory. We found that the maximum value of the probability for quantum theory is $0.0181938$, and for {\em GNST} it is 1/3. Interestingly, for both cases maximal $q$ is lower than for two two-level systems. Whereas, the maximum probability of success of conventional Hardy-type argument (\ref{HardyG}) in {\em GNST} for both two two-level systems and three two-level systems turns out to be $\frac{1}{2}$ \cite{CGKKRR10}.

On the course of preparing this work we have learned about other generalization of Hardy paradox for qubits system \cite{OH}, greatly inspired by our result, presented in \cite{1303.0128}, an arXiv e-print superseded by this manuscript. This shows that manipulating the conditions used in the Hardy paradox allow to choose a class of multi-qubit entanglement detected, but the key feature is to refer to marginal probabilities.
\section{Acknowledgments}
We thank Guruprasad Kar, Sibasish Ghosh and Mohamed Bourennane for stimulating discussions. R.R. acknowledges support
by Foundation for Polish Science (FNP) TEAM/2011-8/9 project co-financed by EU European Regional Development Fund, and ERC grant QOLAPS(291348). M.W. acknowledges support from FNP (HOMING PLUS/2011-4/14) and NCN Grant No. 2012/05/E/ST2/02352. M.\.Z. was supported by project QUASAR.

\end{document}